\documentclass[a4paper,11pt]{article}
\usepackage{graphicx}
\usepackage{here}
\usepackage{amsmath,amssymb,amsthm}
\usepackage{mathrsfs}
\usepackage{ascmac}
\usepackage[sort, comma]{natbib}
\usepackage[subrefformat=parens]{subcaption}
\usepackage{url}
\usepackage{tikz}
\usetikzlibrary{intersections,calc,decorations.pathreplacing}

\newtheorem{theo}{Theorem}
\newtheorem{lem}{Lemma}
\newcommand{\ol}[1]{\overline{#1}}

\title{{\Large
{Agglomeration and welfare of the Krugman model\\in a continuous space}
}}
\author{Kensuke Ohtake\thanks{School of General Education, Shinshu University, Matsumoto, Nagano 390-8621, Japan,
E-mail: k\_ohtake@shinshu-u.ac.jp}}
\date{February 14, 2023}

\begin{document}
\maketitle

\begin{abstract}
Two spatial equilibria, agglomeration, and dispersion, in a continuous space core-periphery model, are examined to discuss which equilibrium is socially preferred. It is shown that when transport cost is lower than a critical value, the agglomeration equilibrium is preferable in the sense of Scitovszky criterion, while when the transport cost is above the critical value, the two equilibria can not be ordered in the sense of Scitovszky criterion.
\end{abstract} 

\noindent
{\bf Keywords:\hspace{1mm}}
compensation principle;
continuous racetrack economy;
core-periphery model;
economic agglomeration;
new economic geography;
self-organization;
transport cost;
welfare economics


\section{Introduction}
\citet{CharGaigNicoThi2006} discuss the welfare economic properties of agglomeration and dispersion equilibrium in the two-regional core-periphery model\footnote{\citet[Section 8.3]{FujiThis} outline \citet{CharGaigNicoThi2006}.}, which has been introduced by \citet{Krug91}\footnote{\citet{PflSu2008} use a quasi-linear utulity model to analyze welfare of two-regional spatial economy.}. One of their main results is that agglomeration can be preferred over dispersion in the sense of Scitovszky criterion \citep{Scit1941} if the transport cost is sufficiently low.

In this paper, we apply the method devised by \citet{CharGaigNicoThi2006} to Krugman's core-periphery model, which is characterized in particular by the CES function and iceberg transport cost\footnote{Our results should depend on these assumptions, which we discuss in Section \ref{cd}. Additionally, the transport cost only for manufacture is also an important assumption. See Footnote 6 at Section \ref{eaw} on this matter.}, in a continuous periodic space\footnote{The core-periphery model on a continuous periodic space has been introduced by \citet[Chapter 6]{FujiKrugVenab}. \citet{OhtaYagi_point} review the behavior of solutions of the model, and discuss its economic implications. \citet{TabaEshiSakaTaka} have onducted a rigorous study of the model on $n$-dimensional continuous space based on real analysis.} and compare the agglomeration and dispersion equilibrium. From the perspective of self-organization, the periodic space assumption is made to focus purely on the self-organizing mechanisms inherent in economy and to eliminate spatial asymmetries caused by some boundary conditions as stated by \citet[pp.22-23]{Krug}. In fact, there are many theoretical studies of the core-periphery model in some periodic or symmetric space such as \citet{Krug93}, \citet{Moss}, \citet{Anas04}, \citet{TabuThis}, \citet{CastCorrMoss}, \citet{AkaTakaIke}, \citet{IkeAkaKon}, \citet{IkeMuroAkaKoTa}, and \citet{GasCasCorr}. Our study is then an attempt to discuss the welfare propertiy, which is free from any specific spacial asymmetries, of the self-organizing mechanism. We prove that even in a periodic continuous space, also similar to \citet{CharGaigNicoThi2006}, agglomeration is preferred in the sense of Scitovszky criterion under sufficiently low transport cost.

Let us describe the model setup we use. The geographical space is assumed to be a circle $S$ of which the radius is $r\geq 1$. The industry consists of two sectors: manufacturing under the monopolistic competition and agriculture under the perfect competition. The manufacturing produces various differentiated goods with the increasing returns, while the agriculture produces one variety of homogeneous good with the constant returns. The supply of each variety of the manufacturing goods is normalized to be $\mu\in(0,1)$ as in \citep[Chapter 4, (4.33)]{FujiKrugVenab}. There are two different types of workers for each sector. Manufacturing workers whose total population is $\mu\in(0,1)$ are assumed to move around $S$ in search of higher real wages, while agricultural workers whose total population is $1-\mu$ are assumed to be unable to move. The transportation of the manufacturing goods incurs the iceberg transport cost. That is, to deliver one unit of any one variety of the manufacturing goods, $T(x,y)\geq 1$ units of it must be shipped from $x\in S$ to $y\in S$. Meanwhile, the transportation of the agricultural good does not incur any transport cost.

The model which we use is described as follows\footnote{The fifth equation of \eqref{1} represents the movement of the manufacturing workers, and $v>0$ denotes the speed of adjustment. However, since comparative statics is the main subject, the differential equation is not explicitly considered in this paper.}.
\begin{equation}\label{1}
\left\{ \begin{aligned}
&Y(t,x) = \mu\lambda(t,x)w(t,x) + (1-\mu)\phi(x),\\
&w(t,x) = \left[\int_S Y(t,y)G(t,y)^{\sigma-1}e^{-(\sigma-1)\tau |x-y|}dy\right]^{\frac{1}{\sigma}},\\
&G(t,x) = \left[\int_S \lambda(t,y)w(t,y)^{1-\sigma}e^{-(\sigma-1)\tau |x-y|}dy
\right]^{\frac1{1-\sigma}},\\
&\omega(t,x) = w(t,x)G(t,x)^{-\mu},\\
&\frac{\partial\Lambda}{\partial t}(t,x)
 = v \left[\omega(t,x)-\int_S \omega(t,y)\lambda(t,y)dy\right]
\lambda(t,x),
\end{aligned}  \right.
\end{equation}
with an initial condition $\lambda(0,x) = \lambda_0(x)$. For any function $f$ on $S$, we denote the integration of $f$ over any subset $\Sigma\subset S$ by $\int_\Sigma f(x)dx$, where $dx$ is a line element\footnote{Names of variables on $S$ can be $x$, $y$, $z$, and so on, as the case may be.}. The functions $Y(t,x), w(t,x), G(t,x)$ and $\omega(t,x)$ represent the income, manufacturing nominal wage, manufacturing price index, and manufacturing real wage at time $t\geq 0$ in region $x\in S$, respectively. Each of the nominal wage of the agricultural workers and the price of the agricultural good is assumed to be one. The manufacturing population density is given by $\mu\lambda(t,x)$ at $t\geq 0$ in $x\in S$, and the function $\lambda$ must satisfy $\int_S\lambda(t,x)dx=1,~\forall t\geq 0$. Similarly, the agricultural population density is given by $(1-\mu)\phi(x)$ in $x\in S$, and the function $\phi$ must satisfy $\int_S\phi(x)dx=1$. The iceberg transportation is assumed by $T(x,y)=e^{\tau|x-y|}$, where $\tau>0$. Here, $|x-y|$ denotes the shorter distance between $x,y\in S$. The parameter $\sigma > 1$ denotes the consumer's preference for manufacturing variety; the closer $\sigma$ is to one, the stronger the degree to which consumers prefer diversity of the manufacturing goods. Since the parameters $\tau$ and $\sigma$ often appear in the form $\tau(\sigma-1)$, it is convenient to introduce 
\[
\alpha :=\tau(\sigma-1)>0
\]
in the following. 

Any $x\in S$ corresponds one-to-one to a certain $\theta\in[-\pi,\pi)$ as $x=x(\theta)$, so it is convenient for specific calculations to use $\theta$ as the coordinates put into $S$. Then, when $x=x(\theta)$ and $y=y(\tilde{\theta})$, the shorter distance $|x-y|$ between them is computed by $\min\left\{r|\theta-\tilde{\theta}|_{\rm abs}, 2\pi r-r|\theta-\tilde{\theta}|_{\rm abs}\right\}$ where $|\cdot|_{\rm abs}$ denotes the absolute value.

\section{Equilibrium and welfare}\label{eaw}
In this section, we first give precise definitions for the two equilibria\footnote{For the two equilibrium, agglomeration and dispersion, it is well known that agglomeration (resp. dispersion) is stable when the transport cost is sufficiently low (resp. high) in the original Krugman model. On the other hand, in models that incorporate more complex cost structures, the dispersion equilibrium could stable not only when the transport cost is high but also when the transport cost is low (re-dispersion). For example, \citet[Chapter 7]{FujiKrugVenab} have given a model that introduces differentiated agriculture and its transport cost (\citet{Ohta2022spa} discusses its  continuous periodic version.), and \citet{MuraThi} and \citet{TakaIkeThi20} have discussed models that incorporate commuting costs.}. We then discuss whether compensation is possible in each equilibrium, i.e., whether each equilibrium is potentially Pareto-dominant over the other.

\subsection{Agglomeration equilibrium}
We first consider the equilibrium in which all the manufacturing workers concentrate in $x^*\in S$. We call this equilibrium (A) for short. It can be expressed as
\begin{equation}\label{deltaonS}
\lambda(x) = \delta^S(x-x^*)
\end{equation}
by using the delta function on $S$. Here, the delta function on $S$ is a linear functional satisfying
\[
\begin{aligned}
&\int_S \delta^S(x-x^*) dx=1,\\
&\int_S \delta^S(x-x^*)f(x)dx = f(x^*),
\end{aligned}
\]
where $f$ is any continuous function on $S$. By substituting \eqref{deltaonS} into the first three equations of \eqref{1}, one can see that
\begin{align}
&w(x^*) =G(x^*) = 1, \label{AwG1}\\
&G(x) = e^{\tau|x-x^*|},~x\in S. \label{Apr}
\end{align}
Then, from the fourth equation of \eqref{1}, the manufacturing real wage in $x^*$ is given by
\begin{equation}
\omega(x^*) = 1. \label{Amr}
\end{equation}

\subsection{Dispersion equilibrium}
We next consider the equilibrium in which the manufacturing workers 
are evenly distributed throughout $S$, and the nominal wages and the price indices of all regions are homogeneous. We call this equilibrium (D) for short. By substituting 
\[
\begin{aligned}
&\lambda(x) \equiv \ol{\lambda} = \frac{1}{2\pi r},\hspace{3.3mm}\forall x\in S,\\
&w(x) \equiv \ol{w},\hspace{14mm}\forall x\in S,\\
&G(x) \equiv \ol{G},\hspace{13.6mm}\forall x\in S.
\end{aligned}
\]
into the first three equations of \eqref{1}, one can obtain 
\begin{align}
&\ol{w} = 1,\nonumber\\
&\ol{G} = \left[\frac{1-e^{-\alpha\pi r}}{\alpha\pi r}\right]^{\frac{1}{1-\sigma}}. \label{Dpr}
\end{align}
Then, from the fourth equation of \eqref{1}, the manufacturing real wage also becomes homogeneous as 
\begin{equation}
\ol{\omega} = \left[\frac{1-e^{-\alpha\pi r}}{\alpha\pi r}\right]^{\frac{\mu}{\sigma-1}}. \label{Dmr}
\end{equation}

\subsection{Comparison of welfare levels}
\subsubsection{Welfare levels at (A)}
Let $\omega^A$, $\psi^A(x)$ denote levels of welfare in (A) of the manufacturing workers and the agricultural workers in $x\in S$, respectively. The level of welfare is measured by real wages of workers. Therefore, from \eqref{Amr} $\omega^A=1$, and from \eqref{Apr} $\psi^A(x)=e^{-\mu\tau|x-x^*|}$.  Note that the agricultural workers at $x^*$ enjoy the same level of welfare as the manufacturing workers because $\psi^A(x^*)=1=\omega^A$. The results are summarized in Table \ref{tab:aggwel}. 

\begin{table}[H]
\centering
  \begin{tabular}{|l|c|}\hline
  \multicolumn{1} {|c|} {Type of workers} & {Welfare} \\ \hline\hline
  Manufacturing workers & $1$\\ \hline
  Agricultural workers at $x^*$ & $1$\\ \hline
  Agricultural workers at $x$ & $e^{-\mu\tau|x-x^*|}$\\ \hline
  \end{tabular}
  \caption{Levels of welfare at (A)}
  \label{tab:aggwel}
\end{table}

\subsubsection{Welfare levels at (D)}
Let $\omega^D$, $\psi^D$ denote levels of welfare in (D), which is measured by the real wages, of the manufacturing workers and the agricultural workers, respectively. From \eqref{Dmr}, $\omega^D=\left[\frac{1-e^{-\alpha r\pi}}{\alpha r \pi}\right]^{\frac{\mu}{\sigma-1}}$, and from \eqref{Dpr}, $\psi^D=\left[\frac{1-e^{-\alpha r\pi}}{\alpha r \pi}\right]^{\frac{\mu}{\sigma-1}}$. Note that all the agricultural workers enjoy the same level of welfare as the manufacturing workers. The results are summarized in Table \ref{tab:diswel}. 

\begin{table}[H]
\centering
  \begin{tabular}{|l|c|}\hline
  \multicolumn{1} {|c|} {Type of workers} & {Welfare} \\ \hline\hline
  Manufacturing workers & $\left[\frac{1-e^{-\alpha r\pi}}{\alpha r \pi}\right]^{\frac{\mu}{\sigma-1}}$\\ \hline
  Agricultural workers & $\left[\frac{1-e^{-\alpha r\pi}}{\alpha r \pi}\right]^{\frac{\mu}{\sigma-1}}$\\ \hline
  \end{tabular}
  \caption{Levels of welfare at (D)}
  \label{tab:diswel}
\end{table}

\begin{theo}\label{omDpsDlt1}
The real wage of the workers in {\rm (D)} is lower than $1$:
\[\omega^D (=\psi^D) < 1.\]
\end{theo}
See Appendix for the proof. 

Therefore, it is immediate that
\[
\begin{aligned}
&\omega^D < 1 = \omega^A,\\
&\psi^D < 1 = \psi^A(x^*).
\end{aligned}
\]
This means that all the manufacturing workers and the agricultural workers in $x^*\in S$ prefer (A) over (D). Meanwhile, as for the agricultural workers in $S\setminus\left\{x^*\right\}$, there exists a neighbor of $x^*$ denoted by $\Gamma\subset S$ such that
\[
\psi^A(x) \geq \psi^D,~\forall x\in \Gamma,
\]
and
\[
\psi^A(x) < \psi^D,~\forall x\notin \Gamma.
\]
See Figure \ref{fig:Gamma} for the situation. 
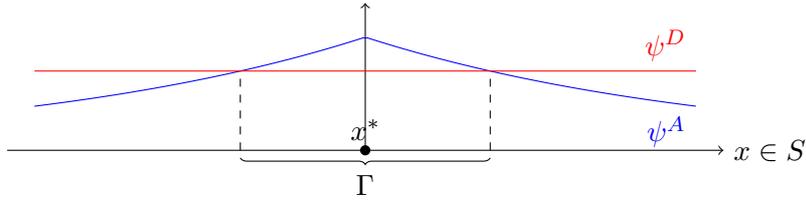
\begin{figure}[H]
\centering
\begin{tikzpicture}[scale=1.5]
 \draw[->,thin] (-3.14,0)--(3.14,0)node[right]{$x\in S$};%x-axis
 \draw[->,thin] (0,0)--(0,1.3)node[right]{$$};%y-axis
 \draw (0,0) node[above]{$x^*$};
 \fill (0,0) circle(0.045);%fill point x*
 \draw[name path=pA,samples=140,blue,thin,domain=-2.9:2.9] plot(\x,{exp(-0.325*abs(\x))})node[below left]{$\psi^A$};
 \draw[name path=pD,samples=2,red,thin,domain=-2.9:2.9] plot(\x,{pow((1-exp(-1.3*3.14))/(1.3*3.14), 0.5/(2.0))})node[above left]{$\psi^D$};
 \path[name intersections={of=pA and pD}];
 \coordinate (L) at (intersection-1);% name left intersection as L
 \coordinate (R) at (intersection-2);% name right intersection as R
 \draw[dashed]($(-3.14,0)!(L)!(3.14,0)$)node[below]{}--(L);
 \draw[dashed]($(-3.14,0)!(R)!(3.14,0)$)node[below]{}--(R);
 \draw[thin,decorate,decoration={brace,mirror,raise=1mm}]($(-3.14,0)!(L)!(3.14,0)$)--($(-3.14,0)!(R)!(3.14,0)$)node[midway,below=2mm]{$\Gamma$};
\end{tikzpicture}
\caption{Sketch of the graph of $\psi^A$ and $\psi^D$}
\label{fig:Gamma}
\end{figure}

\noindent
This means that on the one hand, the agricultural workers in $x\in \Gamma$ prefer (A) over (D); on the other hand, the agricultural workers in $x\notin \Gamma$ prefer (D) over (A).

In the following, we set $x^*=x(0),~ 0\in [-\pi,\pi)$ without loss of generality. Then, $\Gamma\subset S$ corresponds to an interval $[-\Delta, \Delta]\subset [-\pi, \pi]$ where $\Delta>0$. These notations are essential when computing specific integrals.

\subsection{Compensation scheme}

\subsubsection{Determination of nominal compensation}
Under (A), assume that the manufacturing workers and the agricultural workers in $\Gamma$ compensate to the agricultural workers in $S\setminus\Gamma$. Let $C_A(x-x^*)$ denote the compensation received per agricultural workers in $S\setminus\Gamma$. Then, their welfare level after receiving the compensation is given by
\[
\left(1 + C_A(x-x^*)\right)e^{-\mu\tau|x-x^*|},
\]
and it must be equal to their welfare level under (D):
\[
\psi^D = \left[\frac{1-e^{-\alpha r\pi}}{\alpha r\pi}\right]^{\frac{\mu}{\sigma-1}}.
\]
Therefore, the compensation is calculated as
\[
C_A(x-x^*) = \psi^D e^{\mu\tau|x-x^*|} - 1,
\]
and the total compensation is calculated by
\begin{equation}\label{totalcomp}
\int_{S\setminus\Gamma}C_A(x-x^*)(1-\mu)\ol{\phi}dx
=
\int_{S\setminus\Gamma}\left\{\psi^De^{\mu\tau|x-x^*|} - 1\right\}(1-\mu)\ol{\phi}dx
\end{equation}
The manufacturing workers and agricultural workers in $\Gamma$ pay $T_A^*$ and $T_A(x-x^*)$ per person, respectively, to compensate the agricultural workers in $S\setminus\Gamma$. Since the total payment is equal to the total compensation, 
\begin{equation}\label{tteqtc}
\mu T_A^* + (1-\mu)\ol{\phi}\int_\Gamma T_A(x-x^*)dx 
= (1-\mu)\ol{\phi}\int_{S\setminus \Gamma} C_A(x-x^*)dx
\end{equation}
holds.

\subsubsection{Balance of supply and demand.}
The demand from region $x$ for one variety of manufacturing goods produced in region $x^*$ is given by\footnote{\citet[p.50, (4.16)]{FujiKrugVenab}}
\begin{equation}\label{demand}
\mu Y(x)\left(p(x^*)e^{\tau|x-x^*|}\right)^{-\sigma}G(x)^{\sigma-1},
\end{equation}
where $p(x)$ denotes the f.o.b price \citep[p.49]{FujiKrugVenab} of the manufacturing goods produced in $x\in S$. Since $p(x^*)=w(x^*)$ from normalization in \citep[ (4.30)]{FujiKrugVenab}, and we now have \eqref{AwG1}, so \eqref{demand} becomes
\[
\mu Y(x)e^{-\sigma\tau|x-x^*|}G(x)^{\sigma-1}.
\]
Under the assumption of the iceberg transport cost, $e^{\tau|x-y|}$ times this amount has to be shipped to meet this demand. Thus, the aggregate demand for one variety of manufacturing goods produced in region $x^*$ is given by
\begin{equation}\label{aggdemandA}
\begin{aligned}
q(x^*)
&= \int_S \mu Y(x) G(x)^{\sigma-1}e^{-\alpha|x-x^*|}dx.
\end{aligned}
\end{equation}
The income after the compensation is made under (A) is given by
\begin{equation}\label{Yaftercomp}
Y(x)  = \left\{
\begin{aligned}
&\mu\delta^S(x-x^*)\left(1-T_A^*\right)\\
&\hspace{10mm}+(1-\mu)\ol{\phi}\left(1-T_A(x-x^*)\right),\hspace{3mm}\forall x\in \Gamma,\\
&(1-\mu)\ol{\phi}\left(1+C_A(x-x^*)\right),\hspace{17.5mm}\forall x\in S\setminus \Gamma.
\end{aligned}\right.
\end{equation}
Substituting \eqref{Yaftercomp} into \eqref{aggdemandA} yields
\[
\begin{aligned}
q(x^*)&= \mu -\mu^2 T_A^* \\
&\hspace{10mm}-\mu(1-\mu)\ol{\phi}\int_\Gamma T_A(y-x^*)dy\\
&\hspace{10mm}+\mu(1-\mu)\ol{\phi}\int_{S\setminus\Gamma}C_A(y-x^*)dy.
\end{aligned}
\]
Applying \eqref{tteqtc} to this gives
\begin{equation}\label{qsmu}
q(x^*) = \mu
\end{equation}
which means the demand for a variety of the manufacturing goods is equal to its supply $\mu$. Thus, it is possible to implement such compensation without disturbing the equilibrium price.

\subsubsection{Condition under which agglomeration is preferred}
For (A) to be preferred, the following two conditions must be satisfied.
\begin{enumerate}
\item The welfare level of the manufacturing workers in (A) must be higher than that in (D) even after compensation, i.e., 
\begin{equation}\label{1mTsAgtpsiD}
1-T^*_A > \psi^D.
\end{equation}
\item The welfare level of the agricultural workers in $\Gamma$ after compensation must be higher than or equal to that in (D), i.e., 
\begin{equation}\label{agrcmpgth}
\left(1-T_A(x-x^*)\right)e^{-\mu\tau|x-x^*|} \geq \psi^D,\quad\forall x\in \Gamma.
\end{equation}
\end{enumerate}

In the following, we consider a compensation scheme in which the agricultural workers in $\Gamma$ exhaust all the surplus generated by the agglomeration. In this case, \eqref{agrcmpgth} holds as
\begin{equation}\label{agrcmpgtheqal}
T_A(x-x^*) = 1 - \psi^De^{\mu\tau|x-x^*|}.
\end{equation}
Considering \eqref{agrcmpgtheqal}, one can find the total payment by the manufacturing workers and the agricultural workers in $\Gamma$ becomes
\begin{equation}\label{totalpay}
\mu T_A^* 
+ (1-\mu)\ol{\phi}\int_\Gamma \left(1 - \psi^D e^{\mu\tau|x-x^*|}\right)dx 
\end{equation}
Since \eqref{totalpay} equals to the total compensation \eqref{totalcomp}, we see
\[
T_A^*
= \frac{1-\mu}{\mu}\left[\psi^D\cdot\ol{\phi}\int_S e^{\mu\tau|x-x^*|}dx-1\right].
\]

We can show \eqref{1mTsAgtpsiD}, i.e., $F(\tau):=1-T_A^* - \psi^D>0$ holds for sufficiently small values of $\tau>0$. See Appendix for the proof.
\begin{theo}\label{Fgtzero}
There exists $\tau_K > 0$ such that $F(\tau) > 0$ for any $\tau \in (0, \tau_K)$.
\end{theo}
Therefore, we can claim that (A) is preferred over (D) in the sense of Kaldor criterion \citep{Kaldor1939} if and only if $\tau < \tau_K$.

\subsubsection{Compensation under (D) disturbs equilibrium price}\label{subsec:Hikcs_(A)>(D)}
Contrary to the previous discussion, consider compensation under (D) from the agricultural workers in $S\setminus\Gamma$, who gain more at (D) than at (A), to the manufacturing workers and the agricultural workers in $\Gamma$, who prefer (A).

In this case, let $T_D(x-x^*)$ be the amount paid by one agricultural worker in $S\setminus \Gamma$. This amount generally depends on the distance $x-x^*$ from the potential city in $x^*$, since those who live farther away from $x^*$ will suffer greater losses from the agglomeration and will therefore be prepared to pay more to prevent it. 

Meanwhile, let $C_D^*$ and $C_D(x-x^*)$ be the amount of compensation received by a manufacturing worker and an agricultural worker in $\Gamma$, respectively. The latter generally depends on the distance $x-x^*$ from the potential city in $x^*$, since the closer a player lives to $x^*$, the greater the gain he enjoys under (A), and he will not be satisfied unless he receives a larger compensation. 

The income distribution after compensation under (D) is given by
\begin{equation}\label{D_Y}
Y(x)  = \left\{
\begin{aligned}
&\mu\ol{\lambda}(1+C_D^*)+(1-\mu)\ol{\phi}\left(1+C_D(x-x^*)\right),~\forall x\in \Gamma,\\
&\mu\ol{\lambda}(1+C_D^*)+(1-\mu)\ol{\phi}\left(1-T_D(x-x^*)\right),~\forall x\in S\setminus \Gamma.
\end{aligned}\right.
\end{equation}
The total payment and the total compensation must be equal, so
\begin{equation}\label{D_tteqtc}
\begin{aligned}
(1-\mu)\ol{\phi}\int_{S\setminus\Gamma}T_D(x-x^*)dx
=
\mu C_D^* + (1-\mu)\ol{\phi}\int_\Gamma C_D(x-x^*)dx
\end{aligned}
\end{equation}
holds.

Same as \eqref{aggdemandA}, the aggregate demand for one variety of manufacturing goods produced in region $x$ is given by
\begin{equation}\label{aggdemandD}
\begin{aligned}
q(x)
&= \int_S \mu Y(y) G(y)^{\sigma-1}e^{-\alpha|y-x|}dy.
\end{aligned}
\end{equation}
Using the fact that
\[
\ol{G}^{\sigma-1}\int_{S}e^{-\alpha|x-y|}dy = 2\pi r
\]
and substituting \eqref{D_Y} into \eqref{aggdemandD} gives
\begin{equation}\label{qxD}
\begin{aligned}
q(x) 
&= \mu + \mu^2 C_D^* \\
&\hspace{7mm} + \ol{G}^{\sigma-1}\mu(1-\mu)\ol{\phi}\int_\Gamma C_D(y-x^*)e^{-\alpha|x-y|}dy \\
&\hspace{7mm}- \ol{G}^{\sigma-1}\mu(1-\mu)\ol{\phi}\int_{S\setminus\Gamma} T_D(y-x^*)e^{-\alpha|x-y|}dy.
\end{aligned}
\end{equation}
Let us especially consider $q(x^*)$. Then, the integral terms in \eqref{qxD} are estimated as
\begin{equation}\label{estintCD}
\int_\Gamma C_D(y-x^*) e^{-\alpha|x^*-y|}dy
> e^{-\alpha r\Delta}\int_\Gamma C_D(y-x^*) dy
\end{equation}
and 
\begin{equation}\label{estintTD}
\int_{S\setminus\Gamma} T_D(y-x^*) e^{-\alpha|x^*-y|}dy
< e^{-\alpha r\Delta}\int_{S\setminus\Gamma} T_D(y-x^*) dy,
\end{equation}
and note that 
\begin{equation}\label{estGsigminus1gt1}
\ol{G}^{\sigma-1} = \frac{\alpha \pi r}{1-e^{-\alpha \pi r}} > 1.
\end{equation}
These estimations \eqref{estintCD}-\eqref{estGsigminus1gt1} with \eqref{qxD} give
\begin{equation}\label{qxstarfinal}
\begin{aligned}
q(x^*) &> \mu + \mu^2 C_D^* \\
&\hspace{5mm} + \mu(1-\mu)\ol{\phi}e^{-\alpha r\Delta}\int_\Gamma C_D(y-x^*) dy \\
&\hspace{5mm} - \mu(1-\mu)\ol{\phi}e^{-\alpha r\Delta}\int_{S\setminus\Gamma} T_D(y-x^*) dy.
\end{aligned}
\end{equation}
Finally, applying \eqref{D_tteqtc} to \eqref{qxstarfinal}, we obtain
\[
q(x^*) = \mu + (1-e^{-\alpha r\Delta}) \mu^2 C_D^* > \mu
\]
which means excess demand for varieties of the manufacturing goods produced in $x^*$ occurs. Furthermore, since $q(x)$ is continuous as for $x$, there exists a neighborhood of $x^*$ such that the excess demand $q(x)>\mu$ occurs in any $x$ in the neighborhood. 

This result shows that such a compensation scheme is not feasible, since it cannot be implemented without disturbing the prices in (D). Therefore, (A) is preferred over (D) in the sense of Hicks \citep{Hicks1939} under any values of $\tau>0$.  

\section{Conclusion and discussion}\label{cd}
Based on the above discussion, we can confirm as follows that the main result of \citet{CharGaigNicoThi2006} holds for a periodic continuous space model. Subsection \ref{subsec:Hikcs_(A)>(D)} shows that (A) is preferred over (D) in the sense of Hicks for any value of $\tau>0$. Then, Theorem \ref{Fgtzero} claims that if $\tau<\tau_K$, then (A) is preferred over (D); therefore, when $\tau<\tau_K$, (A) is preferred over (D) in the sense of Scitovszky criterion. On the other hand, when $\tau\geq \tau_K$, there is no equilibrium that is preferred in the sense of both Hicks and Kaldor, and thus, in the sense of Scitovszky criterion, the preferability is undetermined.

One may point out that this result depends on specific assumptions on the model, i.e., the CES utility function and  the iceberg transport cost. Although it is difficult to provide a definite answer at this time, if these assumptions are not satisfied, the conclusion of this paper will not hold as it is, or even if it holds, it will be difficult to prove to some extent for the following reasons. 

First, as is well known, in the model with the CES function (CES model), a pro-competitive effect does not work as the number of firms increases, and thus the markup of price is constant. As \citet[p.52]{FujiKrugVenab} state, relaxing these assumption will lead to firms' pricing behavior based on variable markups of prices\footnote{More precisely, the core-periphery model uses the CES function and also assumes the non-strategic behavior of firms that they consider the price index as a given when maximizing profits \citep[p.51, p.52]{FujiKrugVenab}.}. Therefore, in a model with non-CES function (VES model), in the single point city $x^*$ where all firms are concentrated, competition will be more intense and the markup of price will be lower. The resulting deterioration in firms' sales could lead to a decline in the nominal wage that is funded by them. Since this effect could induce a decline in the real wage in $\Gamma$ and thus the real wage can not be enugh to compensate $S\setminus\Gamma$, so it is possible that (A) is not as socially preferred in the VES model as in the CES model\footnote{Meanwhile, if the price reduction through intenese competition is large enough, (A) may be socially preferred as much as or more than in the CES model. In any case, complex effects can arise that are not present in the CES model.}.

Second, when compensating $S\setminus\Gamma$ from $\Gamma$ in (A), owing to the unchanged demand as in \eqref{qsmu} for goods produced by the manufacturing firms located in $x^*$, the compensation does not perturb the equilibrium price. This might be due in large part to the iceberg transport cost. The possible explanation is as follows. The compensation decreases demand from $\Gamma$, while demand from $S\setminus\Gamma$ increases. The iceberg transport cost can then lead the firm to increase production for $S\setminus\Gamma$ to make up for the decrease in production induced by the decrease in demand from $\Gamma$, because under the iceberg transport cost, the firm needs to produce more goods to transport them farther.

\section{Appendix}

\subsection{Proof for Theorem \ref{omDpsDlt1}}

Let $X:=\alpha r\pi> 0$. When $X\geq 1$, it is obvious that 
\begin{equation}\label{1eXX1}
\frac{1-e^{-X}}{X} < 1.
\end{equation}
When $X<1$, by the Maclaurin expansion of $e^{-X}$, we see that \eqref{1eXX1} is equivalent to
\begin{equation*}
\sum_{k=2}^\infty \left\{\frac{X^{k+1}}{(k+1)!}-\frac{X^k}{k!}\right\}
< 0
\end{equation*}
To prove this, we only have to show that for each $k$
\begin{equation}\label{Xkplus1Xk}
\frac{X^{k+1}}{(k+1)!}<\frac{X^k}{k!} 
\end{equation}
holds. Actually, from the assumption $X<1$, we have $X^{k+1}<X^k$ because
\[
X^{k+1}-X^k=X^k\left(X-1\right) < 0.
\]
Hence, together with the fact $(k+1)!>k!$, we obtain \eqref{Xkplus1Xk} for each $k$.
\qed

\subsection{Proof for Theorem \ref{Fgtzero}}
Tedious calculations yield the following Lemmas \ref{lemm:Hs}-\ref{lemm:Is}.
\begin{lem}\label{lemm:Hs}
For any $\tau>0$, the function $\psi^D=\psi^D(\tau)$ given by
\[
\psi^D=\left[\frac{1-e^{-\alpha r\pi}}{\alpha r\pi}\right]^{\frac{\mu}{\sigma-1}},\quad \alpha:=(\sigma-1)\tau
\]
satisfies
\[
\begin{aligned}
&\lim_{\tau\to 0} \psi^D =  1,\\
&\lim_{\tau\to \infty} \psi^D = 0,\\
&\frac{d\psi^D}{d\tau} < 0,\label{H3}\\
&\lim_{\tau\to 0} \frac{d\psi^D}{d\tau} = -\frac{\mu r\pi}{2},\\
&\lim_{\tau\to \infty} \frac{d\psi^D}{d\tau} = 0.
\end{aligned}
\]
\end{lem}

\begin{lem}\label{lemm:Is}
For any $\tau>0$, the function $\chi=\chi(\tau)$ given by
\[
\chi := \ol{\phi}\int_S e^{\mu\tau|x-x^*|}dx = \frac{e^{\mu\tau r\pi}-1}{\mu\tau r\pi}
\]
satisfies
\[
\begin{aligned}
&\lim_{\tau\to 0} \chi =  1,\\
&\lim_{\tau\to \infty} \chi = \infty,\\
&\frac{d\chi}{d\tau} > 0,\\
&\lim_{\tau\to 0} \frac{d\chi}{d\tau} = \frac{\mu r\pi}{2},\\
&\lim_{\tau\to \infty} \frac{d\chi}{d\tau} = \infty.
\end{aligned}
\]
\end{lem}

Lemmas \ref{lemm:dFdtsmall}-\ref{lemm:dFdtlarge} immediately follow from Lemmas \ref{lemm:Hs}-\ref{lemm:Is}.
\begin{lem}\label{lemm:dFdtsmall}
For a sufficiently small $\tau>0$,
\[\frac{dF}{d\tau}(\tau)>0.\]
\end{lem}

\begin{lem}\label{lemm:dFdtlarge}
For a sufficiently large $\tau>0$,
\[\frac{dF}{d\tau}(\tau)<0.\]
\end{lem}

A bit of technical discussion is required to prove Lemma \ref{lemm:d2Fdt2}.
\begin{lem}\label{lemm:d2Fdt2}
For any $\tau>0$,
\[
\frac{d^2F}{d\tau^2}(\tau) < 0.
\]
\end{lem}
\begin{proof}
Let $f^\prime$ denote the first derivative of any function $f$. To estimate the sign of the second derivative $f^{\prime\prime}$ of $f$, it is convenient to use the fact that
\[
(\ln f)^{\prime\prime} = \frac{f^{\prime\prime}f-(f^\prime)^2}{f^2}.
\]
It follows that if $(\ln f)^{\prime\prime}>0$, then $f^{\prime\prime}f>(f^\prime)^2$ holds. Therefore, if $f>0$, then we have $f^{\prime\prime} > 0$.

The second derivative of $F$ is
\begin{equation}\label{Fdd}
F^{\prime\prime}
= -\frac{1-\mu}{\mu}\left(\psi^D\chi\right)^{\prime\prime} - {\psi^D}^{\prime\prime}.
\end{equation}
Tedious calculation shows that $\left(\ln \psi^D\chi\right)^{\prime\prime} > 0$. Since $\psi^D\chi>0$, we have \begin{equation}\label{HIdd}
(\psi^D\chi)^{\prime\prime} >0.
\end{equation}
Another tedious calculation gives $\left(\ln \psi^D\right)^{\prime\prime} > 0$, and we have
\begin{equation}\label{Hdd}
{\psi^D}^{\prime\prime} > 0
\end{equation}
by the same manner. As a result, by considering \eqref{HIdd} and \eqref{Hdd} together with \eqref{Fdd}, we get $F^{\prime\prime} < 0$.
\end{proof}

It is a simple matter to see Lemma \ref{lemm:limitF} holds.
\begin{lem}\label{lemm:limitF}
\leavevmode
\vspace{-1mm}
\begin{enumerate}
\item \[ \lim_{\tau\to 0} F(\tau) = 0, \]
\item \[ \lim_{\tau\to \infty} F(\tau) < 0. \]
\end{enumerate}
\end{lem}

Then Theorem \ref{Fgtzero} obviously follows from Lemmas \ref{lemm:dFdtsmall}-\ref{lemm:limitF}. Figure \ref{fig:F} shows the sketch\footnote{It does not plot the numerically exact values of the function.} of the graph of $F(\tau)$.

\begin{figure}[H]
\centering
\begin{tikzpicture}
 \draw[name path=xaxis,->,thin] (0,0)--(5.0,0)node[above]{$\tau$};
 \draw[->,thin] (0,-2.5)--(0,2.5)node[right]{$$};
 \draw (0,0)node[below right]{O};
 \draw[name path=Omega1,blue,thin,domain=0:4.5] plot(\x,{-0.35*pow(\x,2)+\x})node[right]{$F(\tau)$};
 \path[name intersections={of=Omega1 and xaxis}];
 \fill[black] (intersection-2) circle (0.07) node[above right]{$\tau_K$};
\end{tikzpicture}
\caption{Sketch of the graph of $F(\tau)$}
\label{fig:F}
\end{figure}
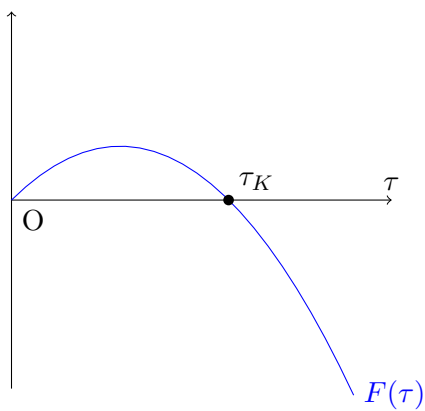

\qed

\bibliographystyle{aer}

\ifx\undefined\bysame
\newcommand{\bysame}{\leavevmode\hbox to\leftmargin{\hrulefill\,\,}}
\fi

\end{document}